\documentclass[review]{elsarticle}
\journal{Journal of \LaTeX\ Templates}

\usepackage{amsmath,amssymb,amsbsy,amsfonts,
latexsym,amsopn,amstext,amsxtra,amsthm}
\usepackage[letterpaper,hmargin=1in,vmargin=1in]{geometry}

\newcommand{\norm}[1]{ \left \lVert #1 \right \rVert_p}
\renewcommand{\L}{\mathcal{L}}
\newcommand{\lat}{\mathcal{L}}
\newcommand{\cL}{\mathcal{L}}
\newcommand{\R}{\mathbb{R}}
\newcommand{\Z}{\mathbb{Z}}
\newcommand{\SVP}{\mathsf{SVP}}
\newcommand{\SIVP}{\mathsf{SIVP}}
\newcommand{\GapSIVP}{\mathsf{GapSIVP}}
\newcommand{\CVP}{\mathsf{CVP}}
\newcommand{\GapCVP}{\mathsf{GapCVP}}
\newcommand{{\GaptwoSAT}}{\mathsf{Gap\text{-}2\text{-}SAT}}
\newcommand{{\GapthreeSAT}}{\mathsf{Gap\text{-}3\text{-}SAT}}

\newcommand{\kSAT}{\mathsf{k\text{-}SAT}}
\newcommand{\GapkSAT}{\mathsf{Gap\text{-}k\text{-}SAT}}

\newcommand{\bt}{\mathbf{t}}
\newcommand{\eps}{\varepsilon}
\newcommand{\spn}{\text{span}}

\newtheorem{theorem}{Theorem}
\newtheorem{defn}{Definition}

\title{A Note on the Concrete Hardness of the Shortest Independent Vector in Lattices}

\author{Divesh Aggarwal}
\address{Department of Computer Science and Centre for Quantum Technologies, NUS}
\ead{dcsdiva@nus.edu.sg}

\author{Eldon Chung}
\address{Department of Computer Science, NUS}
\ead{eldon.chung@u.nus.edu.sg}

\begin{document}

%

\begin{abstract}
Bl{\"o}mer and Seifert~\cite{Blomer:1999:CCS:301250.301441} showed that $\SIVP_2$ is NP-hard to approximate by giving a reduction from $\CVP_2$ to $\SIVP_2$ for constant approximation factors as long as the $\CVP$ instance has a certain property. In order to formally define this requirement on the $\CVP$ instance, we introduce a new computational problem called the Gap Closest Vector Problem with Bounded Minima. We adapt the proof of~\cite{Blomer:1999:CCS:301250.301441} to show a reduction from the Gap Closest Vector Problem with Bounded Minima to $\SIVP$ for any $\ell_p$ norm for some constant approximation factor greater than $1$.

In a recent result, Bennett, Golovnev and Stephens-Davidowitz~\cite{BGS17} showed that under Gap-ETH, there is no $2^{o(n)}$-time algorithm for approximating $\CVP_p$  up to some constant factor $\gamma \geq 1$ for any $1 \leq p \leq \infty$.  We observe that the reduction in~\cite{BGS17} can be viewed as a reduction from $\GapthreeSAT$ to the Gap Closest Vector Problem with Bounded Minima. This, together with the above mentioned reduction, implies that, under Gap-ETH, there is no randomised $2^{o(n)}$-time algorithm for approximating $\SIVP_p$  up to some constant factor $\gamma \geq 1$ for any $1 \leq p \leq \infty$.
\end{abstract}
\maketitle
\section{Introduction}
A lattice $\lat \subset \R^d$ is the set of integer linear combinations  
    \[
    \lat := \lat(\mathbf{B}) = \{z_1 \vec{b}_1 + \cdots + z_n \vec{b}_n \ : \ z_i \in \Z \}
    \]
    of linearly independent basis vectors $\mathbf{B} = (\vec{b}_1,\ldots, \vec{b}_n) \in \R^{d \times n}$. 
 We call $n$ the \emph{rank} of the lattice $\lat$ and $d$ the \emph{dimension} or the \emph{ambient dimension} of the lattice $\cL$. 

For $i = 1, \ldots, n$, the $i^{th}$ successive minimum, denoted by $\lambda_i(\L)$, is the smallest $\ell$ such that there are $i$ non-zero linearly independent lattice vectors that have length at most $\ell$.\par

The Shortest Independent Vector Problem ($\SIVP$) takes as input a basis for a lattice $\lat \subset \R^d$ and $r > 0$ and asks us to decide whether the largest successive minima is at most $r$, i.e., $\lambda_n(\L) \le r$. Typically, we define length in terms of the $\ell_p$ norm for some $1 \leq p \leq \infty$, defined as
\[
\|\vec{x}\|_p := (|x_1|^p + |x_2|^p + \cdots + |x_d|^p)^{1/p}
\]
for finite $p$ and 
\[
\|\vec{x}\|_\infty := \max |x_i|
\; .
\]
We will drop the subscript in $\|\vec{x}\|_p$, when $p$ is clear from the context. 
We write $\SIVP_p$  for $\SIVP$ in the $\ell_p$ norm (and just $\SIVP$ when we do not wish to specify a norm).

Starting with the breakthrough work of Lenstra, Lenstra, and Lov{\'a}sz in 1982~\cite{LLL82}, algorithms for solving lattice problems in both its exact and approximate forms have found innumerable applications, including
factoring polynomials over the rationals~\cite{LLL82}, integer programming~\cite{Lenstra83,Kannan87,DPV11}, cryptanalysis~\cite{Shamir84,Odl90,JS98,NS01}, etc. More recently, many cryptographic primitives have been constructed whose security is based on the (worst-case) hardness of $\SIVP$ or closely related lattice problems \cite{Ajtai96,oded05,GPV08,Pei10,chris_survey}. In particular, the (worst-case) hardness of $\SIVP$ for $\text{poly}(n)$ approximation factors implies the existence of several fundamental cryptographic primitives like one-way functions, collision-resistant hash functions, etc (see, for example, \cite{goldreich1996collision}, \cite{os1998worst}). Such lattice-based cryptographic constructions are likely to be used on massive scales (e.g., as part of the TLS protocol) in the not-too-distant future~\cite{new_hope,frodo,NIST_quantum}. 

 Bl{\"o}mer and Seifert~\cite{Blomer:1999:CCS:301250.301441} showed that $\SIVP$ is NP-hard to approximate for any constant approximation factor. While their result is shown only for the Euclidean norm, their proofs can easily be extended to arbitrary norms. As is true for many other lattice problems, $\SIVP$ is believed to be hard to approximate up to polynomial factors in $n$, the rank of the lattice. In particular, the best known algorithms for $\SIVP$, even for $\text{poly}(n)$ approximation factors run in time exponential in $n$~\cite{ADRS15,ADS15}.

However, NP-hardness itself does not exclude the possibility of sub-exponential time algorithms since it merely shows that there does not exist a polynomial time algorithm unless P = NP.

To rule out such algorithms, we typically rely on a fine-grained complexity-theoretic hypothesis --- such as the Strong Exponential Time Hypothesis (SETH), the Exponential Time Hypothesis (ETH), or the Gap-Exponential Time Hypothesis (Gap-ETH). These  hypotheses were introduced in   \cite{IMPAGLIAZZO2001367}, and are by now quite standard in analyzing the concrete hardness of computational problems. 

To that end, a few recent results have shown quantitative hardness for the Closest Vector Problem ($\CVP_p$)~\cite{BGS17,ABGS19}, and the Shortest Vector Problem ($\SVP_p$)~\cite{AS18} which are closely related. In particular, assuming SETH,~\cite{BGS17,ABGS19} showed that there is no $2^{(1-\eps) n}$-time algorithm for $\CVP_p$ or $\SVP_\infty$ for any $\eps > 0$ and for $1 \leq p \leq \infty$ such that $p$ is not an even integer. Under ETH,~\cite{BGS17} showed that there is no $2^{o(n)}$-time algorithm for $\CVP_p$ for any $1 \leq p \leq \infty$. Also, under Gap-ETH,~\cite{BGS17} showed that there is no $2^{o(n)}$-time algorithm for approximating $\CVP_p$  up to some constant factor $\gamma \geq 1$ for any $1 \leq p \leq \infty$. Similar, but slightly weaker, results were obtained for $\SVP_p$ in~\cite{AS18}. 

\subsection{Our results and techniques.}

 Bl{\"o}mer and Seifert~\cite{Blomer:1999:CCS:301250.301441} showed that $\SIVP_2$ is NP-hard by giving a reduction from $\CVP_2$ to $\SIVP_2$. This reduction can easily be extended to all $\ell_p$ norms, and increases the rank of the lattice by $1$. Thus, combined with the SETH hardness result from~\cite{BGS17,ABGS19}, it implies the following observation.
    \begin{theorem}\label{thm:main1}
        Under the SETH,  there is no $2^{(1-\eps) n}$-time algorithm for $\SIVP_p$ for any $\eps > 0$ and for all $p \ge 1$ such that $p$ is not an even integer. 
    \end{theorem}

A closer look at their reduction reveals that it cannot be extended to showing NP-hardness of approximate $\SIVP$ directly (even though $\CVP$ is known to be NP-hard for almost polynomial approximation factors). The reason for this is that for the lattice $\cL$, when given as a part of a $\CVP$ instance, $\lambda_n(\cL)$ might be much larger than the distance of the target from the lattice, in which case, an oracle for approximating $\SIVP$ up to a constant factor, does not tell anything about the distance of the target from the lattice.

To overcome this difficulty, it was shown in ~\cite{Blomer:1999:CCS:301250.301441} that the $\CVP$ instance obtained from a reduction from the minimum label cover problem has a guarantee that for the CVP instance $(\cL, \mathbf{t})$, $\lambda_n(\cL)$ is ``not much larger" than the distance of $\bt$ from $\cL$.

We introduce a new computational problem called the Gap Closest Vector Problem with Bounded Minima ($\GapCVP^\tau$), which captures the above mentioned requirement on the CVP instance that $\lambda_n(\cL)$ has an upper bound depending on the parameter $\tau$. We observe that the reduction from $\GapthreeSAT$ to $\GapCVP$ in~\cite{BGS17} (which implies  hardness of $\GapCVP$) is actually a reduction from $\GapthreeSAT$ to $\GapCVP^\tau$ for an appropriate choice of $\tau$. We then show a reduction similar to~\cite{Blomer:1999:CCS:301250.301441} from $\GapCVP^\tau$ to $\SIVP$, which implies the following result.

\begin{theorem}
        Under the (randomised) Gap Exponential Time Hypothesis, for any $p \ge 1$, there exists $\gamma' > 1$, $\eps > 0$ such that $\gamma'$-$\SIVP_p$ with rank $n$ is not solvable in $2^{\eps n}$ time.
    \end{theorem}

\section{Preliminaries}
    \subsection{Lattices}   
    Let $\R^n$ be a real vector space, with an $\ell_p$-norm on the vectors such that $\vec{v} \in \R^n, \norm{\vec{v}}^p := \sum_{i=1}^{n} |\vec{v}_i|^p$. 
    
    A lattice $\lat \subset \R^d$ is the set of integer linear combinations  
    \[
    \lat := \lat(\mathbf{B}) = \{z_1 \vec{b}_1 + \cdots + z_n \vec{b}_n \ : \ z_i \in \Z \}
    \]
    of linearly independent basis vectors $\mathbf{B} = (\vec{b}_1,\ldots, \vec{b}_n) \in \R^{d \times n}$. 
 We call $\mathbf{B}$ a basis of the lattice $\lat$, $n$ the \emph{rank} of the lattice,  and $d$ the \emph{dimension} of the lattice $\cL$. 
 If $n = d$, then we say that the lattice is full-rank.\par
    
    Since we wish to have inputs of bounded size, we assume that the coordinates of lattice vectors are rational numbers. Moreover, by appropriately scaling the lattice, we assume without loss of generality, that a $d$-dimensional lattice $\L$ is generated by basis vectors from $\mathbb{Z}^d$. 
    
    \subsection{Successive Minima}
    For a lattice $\L$ of rank $n$, for $1 \le i \le n$, we denote by $\lambda_i^{(p)}(\L)$ the $i^{th}$ successive minimum which is
    the smallest $\ell$ such that there are $i$ linearly independent lattice vectors that have $\ell_p$ norm at most $\ell$. More formally,
\[
    \lambda_i^{(p)}(\lat) := \min\{r \ : \ \dim \left(\spn(\{ \vec{y} \in \lat \ : \ \|\vec{y}\|_p \leq r\}) \right)\geq i\} \; .
\]
    We omit the superscript in $\lambda_i^{(p)}$, when the norm is clear from the context. 
    
    Minkowski's second theorem states the following with regards to the successive minima:\par
    
    \begin{theorem}
        For any $p \ge 1$, and for any full-rank lattice $\L$ we have that \[\left(\prod_{i=1}^n \lambda_i(\L) \right)^\frac{1}{n} \leq n^{\frac{1}{p}}(det(\L))^{\frac{1}{n}}\]
    \end{theorem}
    
    \subsection{Computational problems}
    \textbf{Gap-Closest Vector Problem ($\gamma$-$\GapCVP_p$)}: Given a lattice $\L$, a target vector $\vec{t} \in \mathbb{Z}^n$ (which may or may not be in the lattice) and a value $r>0$, output YES if there exists a vector $\vec{v}$ in the lattice such that $\norm{\vec{v} - \vec{t}\ } \leq r$ (i.e. the closest vector in the lattice to the vector $\vec{t}$ has a distance less than $r$ to the target), and output NO if all the vectors in the lattice have distance greater than $\gamma \cdot r$ to the target.\par
    
    \textbf{Gap-Closest Vector Problem with Bounded Minima ($\gamma$-$\GapCVP_p^\tau$)}: Given a lattice $\L$, a target vector $\vec{t} \in \mathbb{Z}^n$ (which may or may not be in the lattice) with the added guarantee that $\lambda_{n}(\L)^p \leq \tau r^p$, and a value $r$ output YES if there exists a vector $\vec{v}$ in the lattice such that $\norm{\vec{v} - \vec{t}\ } \leq r$ (i.e. the closest vector in the lattice to the vector $\vec{t}$ has a distance less than $r$ to the target), and output NO if all the vectors in the lattice have distance greater than $\gamma \cdot r$ to the target.

    \textbf{Gap-Shortest Independent Vector Problem ($\gamma$-$\SIVP_p$)}: Given a lattice $\L$, and value $r$, output YES if there exists a set of linearly independent vectors $\{\vec{v}_1, \vec{v}_2, ..., \vec{v}_n\}$ that are in $\L$ such that $\max_{i=1}^n \|\vec{v}_i \|_p \le r$, and output NO if for  all such sets, $\max_{i=1}^n \|\vec{v}_i \|_p > \gamma r$.\par

    If we want to talk about the exact variant of these problems (i.e., $\gamma = 1$), then we will omit the prefix $\gamma$.

    $\kSAT$: Given a boolean formula in conjunctive normal form over $n$ variables,  i.e. as a conjunction of $m$ clauses where each clause is a disjunction of $k$ literals, decide if there is an assignment of the $n$ variables such that the boolean formula evaluates to true.\par
    $(\delta, \: \eps)$-$\GapkSAT$: Given a boolean formula in conjunctive normal form with each clause having $k$ literals, and two parameters $0 \leq \delta < \eps \leq 1$, output YES if there exists an assignment such that it satisfies at least $\eps$ fraction of the clauses, and output NO if no assignment satisfies more than $\delta$ fraction of the clauses. For convenience at times the ($\delta$, $\eps$)-prefix may be omitted when it is clear from the context.\par

    \subsection{ETH, SETH and Gap-ETH-hardness}
The following hypotheses were introduced in \cite{IMPAGLIAZZO2001367}, and will be the basis of our hardness results.\par

    \begin{defn}[Exponential Time Hypothesis]
        The Exponential Time Hypothesis (ETH) states that for every $k \geq 3$ there exists a constant $\eps > 0$ such that no algorithm solves $\kSAT$ with $n$ variables in $2^{\eps n}$ time.
    \end{defn}

    \begin{defn}[Strong Exponential Time Hypothesis]
        The Strong Exponential Time Hypothesis (SETH) states that for all $\eps > 0$, there exists a $k \geq 3$ such that no algorithm solves $\kSAT$  with $n$ variables in $2^{(1-\eps)n}$ time.
    \end{defn}
    
    Additionally, \cite{Dinur2016MildlyER} and \cite{DBLP:journals/corr/ManurangsiR16} introduced a ``gap" version of ETH.  The following formulation is from \cite{BGS17}.

    \begin{defn}[Gap Exponential Time Hypothesis]
        There exist constants $\delta < 1$ and $\eps > 0$ such that no algorithm solves $(\delta, 1)$-$\GapthreeSAT$ with $n$ variables in $2^{\eps n}$ time.\par
    \end{defn}
    
        \subsection{Gap-ETH-hardness of $(\delta, \eps)$-$\GaptwoSAT$}
    \begin{theorem}[\cite{GAREY1976237}]
        $\forall \delta, \eps$ such that $0 \leq \delta < \eps \leq 1$, there exists a a polynomial time reduction from $(\delta, \eps)$-$\GapthreeSAT$ with $n$ variables and $m$ clauses to an instance of $(\frac{6 + \delta}{10}$, $\frac{6 + \eps}{10})$-$\GaptwoSAT$, with $n+m$ variables and $10m$ clauses.
    \end{theorem}

    Additionally, Bennett et al. used Dinur's result in \cite{Dinur2016MildlyER} to derive the following result:
    \begin{theorem}[\cite{BGS17}]
        $\forall \delta, \delta'$ such that $0 < \delta < \delta' < 1$, there is a polynomial time-randomised reduction from a $(\delta, 1)$-$\GapkSAT$ with $n$ variables and $m$ clauses, to instances of $(\delta', 1)$-$\GapkSAT$ with $n$ variables and $O(n)$ clauses. 
    \end{theorem}

    This implies it is almost always possible to reduce the number of clauses in $(\delta, 1)$-$\GapkSAT$ instances so that reductions that run linear in $m$ may also be considered linear in $n$, so that Gap-ETH may still apply. However, since the reduction is randomised, existence of sub-exponential time algorithms that solve the resulting instances only imply the existence of randomised sub-exponential time algorithms for $(\delta, 1)$-$\GapkSAT$ in the general case (i.e. when $m = \omega(n))$.

\section{Gap-ETH-hardness of approximating CVP$_p$ with Bounded Minima}




In the following, we show that the reduction from~\cite{BGS17} is in fact a reduction from $\GaptwoSAT$ to $\GapCVP_p^\tau$.

    \begin{theorem}[\cite{BGS17}]
        There exists a reduction from $(\delta, \eps)$-$\GaptwoSAT$ with $n$ variables and $m$ clauses to $\gamma$-$\GapCVP_p^\tau$ for any $p$-norm, so that the rank of the lattice in the resulting instance is the same as the number of variables in the original instance,

        \begin{align*}
    \gamma= \left (\frac{\delta + (1 - \delta) 3^p}{\eps + (1 - \eps) 3^p} \right)^\frac{1}{p} \;,
        \end{align*}

        and

        \begin{align*}
        \tau = \frac{2^p}{\eps + (1-\eps)3^p}
        \end{align*}
    \end{theorem}
\begin{proof}
    We will provide their construction of the $\gamma$-$\GapCVP$ instance, and show that it is actually a $\gamma$-$\GapCVP_p^\tau$ instance. The target vector $\vec{t}$ was defined as:

    \begin{align*}
        t_i &= 3 - \eta_i
    \end{align*}

    where $\eta_i$ denotes the number of negated literals in the $i^{th}$ clause, the distance $r$ was defined as $m^{\frac{1}{p}}(\eps + (1 - \eps)3^p)^\frac{1}{p}$, and the set of basis (column) vectors $\{\vec{b}_1, \vec{b}_2, \dots, \vec{b}_n\}$ was defined as follows.

    \begin{align*}
        b_{i,j} &= \begin{cases}
        \ 2 & \text{if } C_i \text{ contains } x_j  \\
        -2 & \text{if } C_i \text{ contains } \neg  x_j \\
        \ 0 & \text{otherwise}
        \end{cases}
    \end{align*}
    Notice that here, $\vec{b}_j$ for $1 \le j \le n$ is the column vector with $m$ co-ordinates $b_{1,j}, \ldots, b_{m,j}$. 

    In order to prove correctness, we need to show that the resulting instance is indeed an instance of $\GapCVP_p^\tau$. For this, we bound $\lambda_n^p(\lat)$. Clearly, 
    \[
    \lambda_n^p(\lat) \le \max_{j=1^n} \| \vec{b}_j\|^p \le 2^p m \;,
    \]
where we use the fact that each co-ordinate of a basis vector is either $0$, $2$, or $-2$, and hence has absolute value at most $2$. Thus, 
\[
\frac{\lambda_n^p(\lat)}{r^p} \le \tau \;. 
\]

    \end{proof}




\section{Gap-ETH-hardness of approximating SIVP$_p$ within a constant factor}

We now present our main contribution, that is showing hardness of approximating $\gamma$-$\SIVP_p$ within a constant factor $\gamma$. 
    
    \begin{theorem}
        For any $p \ge 1$, $\tau = \tau(n) > 0$ with a polynomial size representation, and any $\gamma \ge 1$, there exists an efficient reduction from $\gamma$-$\GapCVP_p^\tau$ to $\gamma'$-$\GapSIVP_p$ for any $\gamma' \ge 1$ such that \[
        {\gamma'}^p< \frac{r^p + \gamma^p \alpha^p}{r^p + \alpha^p}\;,\] where
        \[
        \alpha^p = \max\left(r^p(\tau - 1)\: , \: \frac{\gamma^p r^p}{2^p -1}\right)\;.\] 
        Moreover, the rank of the lattice in the $\gamma'$-$\GapSIVP_p$ instance is equal to $n+1$ where $n$ is the rank $\gamma$-$\GapCVP_p^\tau$ instance.
    \end{theorem}

    \begin{proof}
        
        Let $(\mathcal{L}, \vec{t}, r)$ denote the given $\gamma$-$\GapCVP_p^\tau$ instance, where $\lat$ is a rank $n$ lattice and $\vec{b}_1, \vec{b}_2, \ldots, \vec{b}_n$ are the basis vectors for $\mathcal{L}$. We will construct a $\gamma'$-$\SIVP_p$ instance  $(\mathcal{L}', r')$. Let $\lambda_n = \lambda_n(\mathcal{L})$, and $\lambda_{n+1}' = \lambda_{n+1}'(\mathcal{L}')$.

        Given a basis for the $\gamma$-$\GapCVP_p^\tau$ instance as $\vec{b}_1, \vec{b}_2, \dots, \vec{b}_n$ and the target vector $\vec{t}$, the reduction constructs the basis $B'$ for $\mathcal{L}'$ given by the column vectors of the  matrix
        \begin{align*}
            \begin{bmatrix}
            \vec{b}_1 & \vec{b}_2 & \ldots & \vec{b}_n & \vec{t}\\
            0   &  0  &  \ldots  &  0  & \alpha\\
            \end{bmatrix}\;.
        \end{align*}
        
Furthermore, the reduction chooses $r' = (r^p + \alpha^p)^{1/p}$. This reduction clearly runs in polynomial time, provided that $\alpha$ does not need too many bits to be represented -- polynomial in the size of the original instance.
We now argue correctness of the reduction.
        
        Let $\vec{v}$ be the vector closest to the target $\vec{t}$, and let $\vec{v}_1, \ldots, \vec{v}_n$ be a set of linearly independent vectors in $\cL$ such that 
\[
\lambda_n = \max(\| \vec{v}_1\|, \ldots, \| \vec{v}_n\|)\;.
\]

Notice that $\vec{v}_1, \ldots, \vec{vv}_n, (\vec{v} - \vec{t}, \alpha)^T$ is a set of linearly independent vectors in $\cL'$. Thus, if the $\gamma$-$\GapCVP_p^\tau$ instance is a YES instance, then 
\[
\lambda_{n+1}' \le \max(\|\vec{v}_1\|, \ldots, \|\vec{v}_n\|, \|\vec{v} - \vec{t}\|) \le \max(\lambda_n, (r^p + \alpha^p)^{1/p})  \;.
\]

Also, any set of linearly independent vectors must have at least one vector which, when written as an integer combination of vectors in $B'$, has a non-zero co-efficient for the last basis vector $(\vec{t}, \alpha)^T$. Let this vector be $\vec{x}$.  So, if the $\gamma$-$\GapCVP_p^\tau$ instance is a NO instance, then if the coefficient of $(\vec{t}, \alpha)^T$ in $\vec{x}$ is $1$ or $-1$, then the length of the vector is at least $(\gamma^p \cdot r^p + \alpha^p)^{1/p}$, and if the coefficient has absolute value at least $2$, then the $n+1$-th coordinate, and hence $\|\vec{x}\|$, is at least $2\alpha$. 

From this, we obtain that if the given instance is a YES instance, then
\[
\lambda_{n+1}'^{p} \leq \max(r^p \cdot \tau, \; r^p + \alpha^p)  = r^p + \alpha^p\;,\]
and hence the reduction outputs YES.
If the given instance is a NO instance, then
        \[\lambda_{n+1}'^{p} \ge \min(\gamma^p r^p + \alpha^{p}, \; 2^p\alpha^p) = \gamma^p r^p + \alpha^{p} \;,
        \]
        and hence the reduction outputs NO. The correctness follows.





\end{proof}

    \begin{theorem}
        Under the randomised Gap Exponential Time Hypothesis, there exists $\gamma' > 1$, $\eps > 0$ such that $\gamma'$-$\GapSIVP_p$ with rank $n$ is not solvable in $2^{\eps n}$ time.
    \end{theorem}

    \begin{proof}
        This can be achieved by considering the instances throughout the chain of reductions from $(\delta, \eps)$-$\GapthreeSAT$ to $(\delta', \eps')$-$\GaptwoSAT$ to $\gamma$-$\GapCVP^\tau_p$ and finally $\gamma'$-$\GapSIVP_p$.

        In the original $(\delta, \eps)$-$\GapthreeSAT$ instance with $n$ variables and $m$ clauses, we obtain a $\gamma'$-$\GapSIVP_p$ with rank $n + m + 1$ with high probability. Thus under the randomised Gap-ETH, there is no sub-exponential time algorithm for $\gamma'$-$\GapSIVP_p$, for all $p \in [1, \infty)$.

    \end{proof}

    \bibliography{Complexity-Of-SIVP,SETHSIVP}{}
\end{document}